\RequirePackage{fix-cm}
\documentclass[smallextended]{svjour3}       
\smartqed  
\usepackage{comment}
\usepackage{graphicx}
\usepackage{amssymb}
\usepackage{amsmath}
\usepackage{etoolbox}
\usepackage{amssymb}
\usepackage{graphicx}
\usepackage[utf8]{inputenc}
\usepackage{amssymb,amsmath,latexsym}
\usepackage{bbm}
\usepackage{graphicx}
\usepackage{xcolor}
\usepackage{graphicx}
\usepackage{calrsfs}
\graphicspath{ {plots/} }
\usepackage[utf8]{inputenc}
\usepackage[english]{babel}

\usepackage{bbold}
\usepackage{graphicx}
\usepackage{subcaption}
\usepackage[cal=cm]{mathalfa}
\usepackage{booktabs} 
\usepackage{hyperref}

\usepackage{hyperref}
\usepackage{algorithm}

 \newcommand\subsetsim{\mathrel{%
  \ooalign{\raise0.2ex\hbox{$\subset$}\cr\hidewidth\raise-0.8ex\hbox{\scalebox{0.9}{$\sim$}}\hidewidth\cr}}}

\usepackage[none]{hyphenat}

\date{Date of Submission: 26/04/2021 }

\author{Ioannis Kordonis$^1$ \and Athanasios-Rafail Lagos $^1$\and George  P. Papavassilopoulos $^{1,2}$  }

\institute{$^1$  National Technical University of Athens, School of Electrical and Computer Engineering, 9 Iroon Polytechniou str., Athens, Postal Code 157 80, Greece \newline $^2$ University of
Southern California, Department of Electrical
Engineering--Systems, 3740 McClintock
Ave, Los Angeles, CA 90089, United States\\
I. Kordonis \email{jkordonis1920@yahoo.com}, A.-R. Lagos
              \email{ lagosth@mail.ntua.gr}, G.P. Papavassilopoulos
              \email{yorgos@netmode.ntua.gr}
}

 \title{Dynamic Games of Social Distancing during an Epidemic: Analysis of Asymmetric Solutions }

\titlerunning{Dynamic Games of Social Distancing during an Epidemic}
\authorrunning{ I. Kodonis, A.-R. Lagos and G.P. Papavassilopoulos}

\begin{document}

\maketitle
\begin{abstract}
\sloppy
Individual behaviors play an essential role in the dynamics of transmission of infectious diseases, including COVID--19. This paper studies a dynamic game model that describes the social distancing behaviors during an epidemic, assuming a continuum of players and individual infection dynamics. The evolution of the players' infection states follows a variant of the well-known SIR dynamics. We assume that the players are not sure about their infection state, and thus they choose their actions based on their individually perceived probabilities of being susceptible, infected or removed. The cost of each player depends both on her infection state and on the contact with others. We prove the existence of a Nash equilibrium and characterize Nash equilibria using nonlinear complementarity problems. We then exploit some monotonicity properties of the optimal policies to obtain a reduced-order characterization for Nash equilibrium and reduce its computation to the solution of a low-dimensional optimization problem. It turns out that, even in the symmetric case, where all the players have the same parameters, players may have very different behaviors. We finally present some numerical studies that illustrate this interesting phenomenon and investigate the effects of several parameters, including the players' vulnerability, the time horizon, and the maximum allowed actions,  on the optimal policies and the players' costs.


\end{abstract}

\section{Introduction}
\sloppy
COVID--19 pandemic is one of the most important events of this era. Until early April 2021, it has caused more than 2.8 million deaths, an unprecedented economic depression, and affected most aspects of people's lives in the larger part of the world. During the first phases of the pandemic, Non--Pharmaceutical Interventions (primarily social distancing) has been one of the most efficient tools to control its spread \cite{european2020guidelines}. Due to the slow roll-out
of the vaccines, their uneven distribution, the emergence of SARS-CoV-2 variants, age limitations, and people's resistance to
vaccination, social distancing is likely to remain significant in large part of the globe for the near future.

Mathematical modeling of epidemics dates back to early twentieth century  with the seminal works of Ross \cite{ross} and Kermack and McKendrick \cite{kermack}. A widely used modeling  approach separates people in several compartments according to their infection state (e.g. susceptible, exposed, infected, recovered etc.) and derive differential equations describing the evolution of the population of each compartment (for a review see \cite{allen2008mathematical}).  Individual behaviors are essential to the description of the spread of epidemics. Thus, several  game theoretic models were developed,  to study voluntary vaccination \cite{Zhang1,Chang,Bauch1,Bauch2,Reluga1,Reluga2,Zhang2,Fine-Clarkson} and   voluntary implementation of Non-Pharmaceutical Interventions (NPIs) \cite{Reluga3,Poletti2,Poletti3,Kremer,Vardavas,Del_Valle,Chen2,Funk-Review,Funk1,Chen1,d'Onofrio}. Another closely related stream of research is the study of the adoption of decentralized protection strategies in engineered and social networks \cite{theodorakopoulos2012selfish,trajanovski2015decentralized,hota2019game,huang2019differential}.
Recently, with the emergence of the COVID--19 pandemic, there is a renewed interest in modeling individual behaviors. Related tools include dynamic game analysis of social distancing \cite{toxvaerd2020equilibrium,lee2021controlling,cho2020mean,tembine2020covid,aurell2020optimal,elie2020contact}, evolutionary game theory \cite{karlsson2020decisions,amaral2020epidemiological,ye2020modelling,kabir2020evolutionary} and  network game models \cite{lagos2020games,amini2020epidemic}.

This paper presents a game-theoretic model to describe the social distancing choices of individuals during an epidemic. Each player has an infection state, which can be Susceptible (S), Infected (I), or Removed (R). The probability that a player is at each health state evolves dynamically depending on the player's distancing behavior, the others' behavior, and the prevalence of the epidemic. We assume that the players care both about their health state and about maintaining their social contacts. The players may have different characteristics, e.g., vulnerable vs. less vulnerable, or care differently about maintaining their social contacts.

We assume that the players are not sure about their infection state, and thus they choose their actions based on their individually perceived probabilities of being susceptible, infected or removed. In contrast with most of the literature, in the current work, players -- even players with the same characteristics -- are allowed to behave differently. We first characterize the optimal action of a player, given the others' behavior, and show some monotonicity properties of optimal actions. We then prove the existence of a Nash equilibrium and characterize it in terms of a nonlinear complementarity problem.

Using the monotonicity of the optimal solution, we provide a simple reduced-order characterization of the Nash equilibrium in terms of a nonlinear programming problem. This formulation simplifies the computation of the equilibria drastically. Based on that result, we performed numerical studies, which verify that players with the same parameters may follow different strategies. This phenomenon seems realistic since people facing the same risks or belonging to the same age group often have different social distancing behaviors.

The rest of the paper is organized as follows. Section \ref{Sec_Model} presents the game theoretic model. In section \ref{Sec_Analysis}, we analyze the optimization problem of each player and prove some monotonicity properties. In Section  \ref{Sec_NEchar}, we prove the existence of the equilibrium and provide Nash equilibrium characterizations. Section  \ref{Sec_Num} presents some numerical results. Finally, the Appendix contains the proof of the results of the main text.

\section{The Model}
\label{Sec_Model}

This section presents the dynamic model for the epidemic spread and the social distancing game among the members of the society. 

We assume that the infection state of each agent could be Susceptible (S), Infected (I), Recovered (R), or Dead (D). A susceptible person gets infected at a rate proportional to the number of infected people she meets with.  An infected person either recovers or dies at constant rates which depend on her vulnerability. An individual who has recovered from the infection is immune, i.e., she could not get infected again.
The evolution of the infection state of an individual is shown in Figure \ref{markov_proc}.
\begin{figure}[h!]
\centering
  \includegraphics[width=0.35\textwidth]{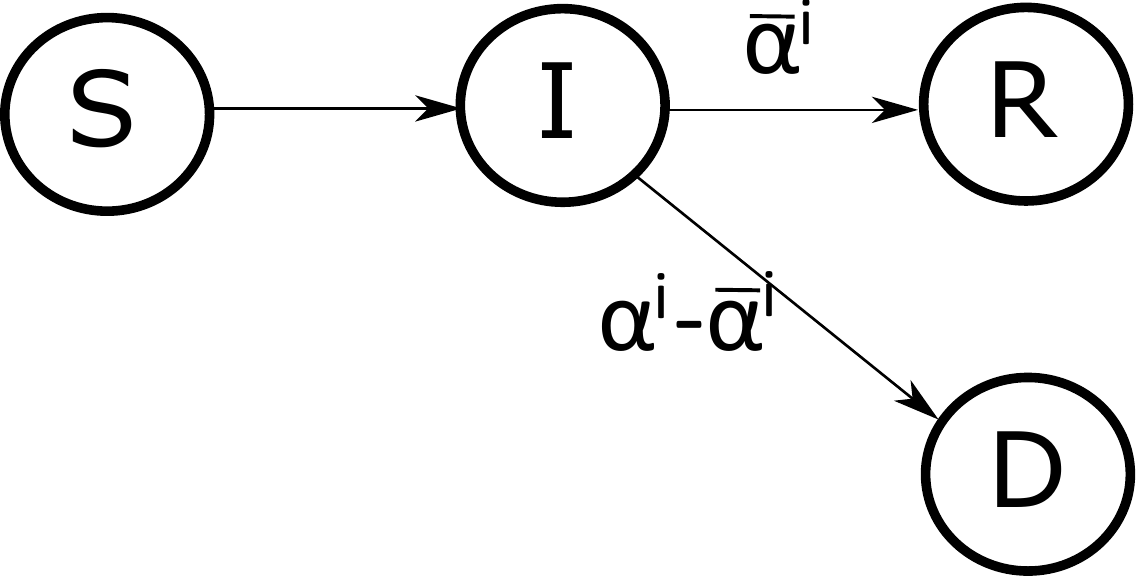}
  \caption{\textit{The  evolution of the infection state of each individual. }}
  \label{markov_proc}
\end{figure}

We assume that there is  a continuum of agents. This approximation is frequently used in  game--theoretic models dealing with a very large number of agents. The set of players is described by the measure space $([0,1),\mathcal B, \mu)$, where $\mathcal B$ is the Borel $\sigma$-algebra and $\mu$ the Lebesgue measure. That is, each player is indexed by an $i\in[0,1)$.

Denote by $S^i(t)$ the probability that player $i\in[0,1)$ is susceptible at time $t$  and by $I^i(t)$ the probability player $i$ is infected.
 The dynamics is given by:
\begin{equation}\label{individual_dynamics}
\begin{aligned}
\dot {S}^i &= -ru^iS^iI^f\\
\dot {I}^i &=ru^iS^iI^f-\alpha^i I^i\\
\end{aligned},
\end{equation}
where $r,\alpha^i$ are positive constants, and $u^i(t)$ is the action of player $i$ at time $t$. The quantity $u^i(t)$ describes player $i$'s socialization, which is proportional to the time she spends in public places. The quantity $I^f$, which denotes the density of infected people in public places, is given by:
\begin{equation}
I^f(t)=\int I^i(t)u^i(t)\mu(di).\label{I_F_def}
\end{equation}
For the actions of the players, we  assume that there are positive constants $u_m$, $u_M$, such that $u^i(t)\in [u_m,u_M]\subset [0,1]$. The constant $u_m$ describes the minimum social contacts needed for an agent to survive and $u_M$ is an upper bound posed by the government.

The cost function for player $i$ is given by:
\begin{equation}\label{payoffs}
  J^i=G^i(1-S^i(T))-s^i\int_{0}^{T}u^i(t)\tilde{u}(t)dt-s^i\int_{0}^{T}\kappa u^i(t)dt,
\end{equation}
where $T$ is the time horizon.
 The first term of \eqref{payoffs} corresponds to the disutility  a player experiences if she gets infected and the parameter $G^i>0$ depends on the vulnerability of the player.  The second term corresponds to the utility a player derives from the interaction with the other players, whose mean action is denoted  by $\tilde u(t)$:
\begin{equation}
\tilde u(t) =\int u^i(t) \mu(di). \label{u_tilde} 
\end{equation}
Finally, the third term indicates the interest of a person to  go outside. The relative magnitude of this desire is modeled by a positive constant $\kappa$.

Considering the auxiliary variable $\bar u(t)$:
\begin{equation}
\bar u(t) =\kappa+\tilde{u}, \label{u_bar_def}
\end{equation}
and computing $S(T)$ by solving \eqref{individual_dynamics}, the cost can be written equivalently as:
\begin{equation}\label{nonlinear_costs}
  J^i =G^i\left(1-S^i(0)e^{-r\int_{0}^{T}u^i(t)I^f(t)dt}\right)-s^i\int_{0}^{T}u^i(t)\bar{u}(t)dt.
\end{equation}

\textbf{Assumption 1}\textit{(Finite number of types):} There are $M$ types of players. Particularly, there are $M+1$ values $0=\bar i_0<\dots<\bar i_M=1$ such that the functions $S^i(0),I^i(0),G^i,s^i,\alpha^i:[0,1)\rightarrow \mathbb R$ are
constant for $i\in [\bar i_0,\bar i_1)$, $i\in[\bar i_1,\bar i_2),\dots$,  $i\in[\bar i_{M-1},\bar i_{M}) $. Denote by $m_j=\mu([\bar i_{j-1},\bar i_{j}))$ the mass of the players of type $j$. Of course $m_1+\dots+m_M=1.$

\begin{remark}
The finite number of types assumption is very common in many applications dealing with a large number of agents. For example, in the current COVID--19 pandemic, people are grouped based on their age and/or underlying diseases to be prioritised for vaccination.  Assumption 1, combined with some results of the following section, is convenient to describe the evolution of the states of a continuum of players using a finite number of differential equations.
\end{remark}

\textbf{Assumption 2}\textit{(Piecewise--constant actions):} The interval $[0,T)$ can be divided in subintervals $[t_k,t_{k+1})$, with $t_0=0<t_1<\dots<t_N=T$, such that the actions of the players are constant in these intervals.

\begin{remark}
Assumption 2 indicates that people decide only a finite number of times ($t_k$) and follow their decisions for a time interval $[t_k,t_{k+1})$.  A reasonable length for that time interval could be  1 week.
\end{remark}

The action of player $i$ in the interval $[t_k,t_{k+1})$ is denoted by $u^i_k$. 

\textbf{Assumption 3}\textit{(Measurability of the actions):}
The function $u^\cdot_k:[0,1)\rightarrow [u_m,u_M]$ is measurable. 

Under Assumptions 1--3, there is a unique solution to differential equations \eqref{individual_dynamics}, with  initial conditions $S^\cdot (0),I^\cdot(0)$,  and the the integrals in \eqref{I_F_def}, \eqref{u_tilde} are well-defined (see Appendix \ref{ProofExistenceDE}). We use the following notation: 
$$\bar u_k=\int_{t_k}^{t_{k+1}} \bar u  dt, \text{~~~ and~~~~} I^f_k = \int_{t_k}^{t_{k+1}} I^f(t) dt.$$ 



For each player we define an auxiliary cost, by dropping the fixed terms of \eqref{nonlinear_costs} and dividing by $s^i$:
\begin{equation}
  \tilde J^i(u^i)=-b^i\exp\left[-r\sum_{k=0}^{N-1}u^i_kI^f_k\right]-\sum_{k=0}^{N-1}u^i_k\bar{u}_k,
  \label{AuxilCost}
\end{equation}
where $b^i=S^i(0)G^i/s^i$, and $u^i=[u^i_0,\dots,u^i_{N-1}]^T$. Denote by $U=[u_m,u_M]^N$, the set of possible actions for each player. Observe that $u^i$ minimizes $J^i$ over the feasible set $U$ if and only if it minimizes the auxiliary cost $\tilde J^i$. Thus, the optimization problem for player $i$ is equivalent to:
\begin{align}
\underset{u^i\in U}{\text{minimize}}~\tilde J^i(u^i).\label{aux_opt_pr}
\end{align}

\textbf{Assumption 4:} For a player $i$ of type $j$ denote $b_j=b^i$. Assume that  the different types of players have different $b_j$'s. Without loss of generality assume that $b_1<b_2<\dots<b_M.$

\textbf{Assumption 5:} Each player $i$ has access only to the probabilities $S^i$ and $I^i$ and the aggregate quantities $\bar u$ and $I^f$, but not the actual infection states.

\begin{remark}     This assumption  is reasonable in cases where the test availability is very sparse, so the agents are not able to have a reliable feedback for their estimated health states.
\end{remark}
In the rest of the paper we suppose that Assumptions 1--5 are satisfied.

\section{Analysis of the Optimization Problem of Each Player}
\label{Sec_Analysis}
In this section, we analyze the optimization problem for a representative player $i$, given $\bar u_k$ and $I^f_k>0$, for $k=0,\dots,N-1$.

Let us first define the following composite optimization problem:
\begin{align}
\underset{A}{\text{minimize}}\left\{ -b^ie^{-A}+ f(A) \right\},
\label{Probl2}
\end{align}
where:
\begin{align}
f(A)=\inf_{u^i\in U}\left\{ -\sum_{k=0}^{N-1}u^i_k\bar{u}_k:  \sum_{k=0}^{N-1}u^i_kI^f_k=A/r  \right\}.
\label{f_def}
\end{align}

The following proposition proves that  \eqref{aux_opt_pr} and \eqref{Probl2} are equivalent and expresses their solution in a simple threshold form.

\begin{proposition}
\label{EquProp}
\begin{itemize}
\item[(i)] If $u^i$ is optimal for \eqref{aux_opt_pr}, then $u^i\in \tilde U=\{u_m,u_M\}^N$.
\item[(ii)] Problems \eqref{aux_opt_pr} and \eqref{Probl2} are equivalent, in the sense that they have the same optimal values, and $u^i$ minimizes \eqref{aux_opt_pr} if and only if there is an optimal $A$ for \eqref{Probl2} such that $u^i$ attains the minimum in \eqref{f_def}.
\item[(iii)] Let $A_m=ru_m\sum_{k=0}^{N-1}I^f_k $ and $A_M=ru_M\sum_{k=0}^{N-1}I^f_k $. For $A\in [A_m,A_M]$, the function $f$ is continuous, non-increasing, convex and piecewise affine. Furthermore, it has at most $N$ affine pieces and $f(A)=\infty$, for $A\not \in[A_m,A_M]$.
\item[(iv)] There are at most $N+1$ vectors $u^i\in U$ that minimize \eqref{aux_opt_pr}.
\item[(v)] If $u^i$ is optimal for  \eqref{aux_opt_pr}, then there is a $\lambda'$ such that $\bar u_k/I^f_k\leq \lambda'$ implies $u^i_k=u_m$, and $\bar u_k/I^f_k>\lambda'$ implies $u^i_k=u_M$.
\end{itemize}
\end{proposition}
\begin{proof} See Appendix \ref{ProofEquProp}\end{proof}

\begin{remark}
The fact that the optimal value of a linear program is a convex  function of the constraints constants is known in the literature (e.g., see \cite{Shapiro} chapter 2).	 Thus, the convexity of the function $f$ is already known from the literature.
\end{remark}

\begin{corollary}
\label{StrategyCorollary}
There is a simple way to solve the optimization problem \eqref{aux_opt_pr} using the following steps:
\begin{itemize}
\item[1.] Compute  $\Lambda = \{\bar u_k/I^f_k: k=0,\dots,N-1\}\cup\{0\}$.
\item[2.] For all $\lambda'\in \Lambda $ compute  $u^{\lambda'} $ with: $$u^{\lambda'}_k = \begin{cases} u_M ~~~\text{if ~}  \bar u_k/I_k^f> \lambda'\\
 u_m ~~~\text{if ~}  \bar u_k/I_k^f\leq \lambda' \end{cases},$$
and $J^i(u^{\lambda'})$.
\item[3.] Compare the values of $J^i(u^{\lambda'})$, for all $\lambda'\in \Lambda$ and choose the minimum.
\end{itemize}
\end{corollary}

We then prove some monotonicity properties for the optimal control.

\begin{proposition} \label{MonotProp}
Assume that for two players $i_1$ and $i_2$, with parameters $b^{i_1}$ and $b^{i_2}$,  the minimizers of \eqref{Probl2} are  $A_1$ and $A_2$ respectively and $u^{i_1}$ and $u^{i_2} $ are the corresponding optimal actions. Then:
\begin{itemize}
\item[(i)] If  $b^{i_1}<b^{i_2}$, then $A_1\geq A_2$.
\item[(ii)] If $b^{i_1}<b^{i_2}$, then $u_k^{i_2}\leq u_k^{i_1}$, for  $k=0,\dots,N-1$.
\item[(iii)] If   $b^{i_1}=b^{i_2}$,
then either $u_k^{i_2}\leq u_k^{i_1}$ for all $k$, or $u_k^{i_1}\leq u_k^{i_2}$ for all $k$.
\end{itemize}
\end{proposition}
\begin{proof} See Appendix \ref{ProofMonotProp}.
\end{proof}

\begin{remark}
Proposition \ref{MonotProp}.(ii) expresses of the fact that if (a) a person is more vulnerable, i.e., she has large $G^i$, or (b) she derives less utility from the interaction with the others, i.e., she has smaller $s^i$, or (c) it is more likely that she is not yet infected, i.e., she has larger $S^i(0)$, then she interacts less with the others. It is probably interesting that small differences in $S^i(0)$ can be amplified.
\end{remark}

\section{Nash Equilibrium Existence and Characterization}
\label{Sec_NEchar}
 \subsection{Existence and NCP characterization}
In this section, we prove the existence of a Nash equilibrium and characterize it in terms of a Nonlinear Complementarity Problem (NCP).

We consider the set $\tilde U=\{u_m,u_M\}^N$, defined in Proposition \ref{EquProp}. Let $v_1,\dots,v_{2^N}$ be the members of the set  $\tilde U$, and $p^j_l$ be the mass  of players of type $j$ following action $v_l \in\tilde U$. Let also $p^j=[p^j_1,\dots,p^j_{2^N}]$ be the distribution  of actions of the players of type $j$ and $\pi=[p^1,\dots,p^M]$  be the distribution of the actions of all the players. 

Denote by:
\begin{equation}
\Delta_j = \{p^j\in \mathbb R^{2^N}: p^j_l\geq 0, \sum_{l=1}^{2^N} p^j_l= m_j \},
\end{equation} 
the set of possible distributions of actions of the players of type $j$ and by $\Pi=\Delta_1\times\dots\times \Delta_M$ the set of all possible distributions.

Finally, let $F:\Pi\rightarrow \mathbb R^{2^N\cdot M}$ be the vector function of auxiliary costs, that is, the component $F_{(j-1)2^N+l}(\pi)$ is  the auxiliary cost of the players of type $j$ playing a strategy $v^l$, as introduced in \eqref{AuxilCost}, when the distribution of actions is $\pi$. We denote $F^j(\pi)=[F_{(j-1)2^N+1}(\pi),...,F_{j2^N}(\pi)]$ the vector of the auxiliary costs of the players of type $j$ playing $v^l$, $l=1,\dots,2^N$.

Let us recall the notion of a Nash equilibrium  for games with a continuum of players (e.g. \cite{Mas_Colell}).
\begin{definition}
A distribution of actions $\pi\in\Pi$ is a Nash equilibrium if for all $j=1,\dots,M$ and $l=1,\dots,2^N$:
\begin{equation}
\pi_{(j-1)2^N+l}>0 \implies l\in \underset{l'}{\arg\min}~F_{(j-1)2^N+l'}(\pi)
\end{equation}
\end{definition}

Let $\delta^j(\pi)$ be the value of problem \eqref{aux_opt_pr}, i.e., the minimum value of the auxiliary cost of an agent of type $j$. This value depends on $\pi$, through the terms $I^f$ and $\bar u$.
Define $\Phi^j(\pi)=F^j(\pi)-\delta^j(\pi)$ and $\Phi(\pi)=[\Phi^1(\pi)...\Phi^M(\pi)]$.
We then characterize a Nash equilibrium in terms of a Nonlinear Complementarity Problem (NCP):
\begin{equation}\label{NCP_large}
0\leq \pi \perp \Phi(\pi)\geq 0,
\end{equation}
where $\pi \perp \Phi(\pi)$ means that $\pi^T\Phi(\pi)=0$.

\begin{proposition}\label{NCP_Prop}\begin{itemize}
\item[(i)]
 A distribution $\pi\in \Pi$ corresponds to a Nash equilibrium if and only if it satisfies the NCP \eqref{NCP_large}.
 \item[(ii)] A distribution $\pi\in \Pi$ corresponds to a Nash equilibrium if and only if it satisfies the variational inequality:
\begin{equation}(\pi'-\pi)^TF(\pi)\geq 0, \text{~~for all } \pi'\in\Pi\label{VI_large}\end{equation}
\item[(iii)] There exists a Nash equilibrium.
\end{itemize}
\end{proposition}
\begin{proof} See Appendix \ref{ProofNCP_Prop}.
\end{proof}

\begin{remark}
In principle, we can use algorithms for NCPs to find Nash equilibria. The problem is that the number of decision variables grows exponentially with the number of decision steps. Thus, we expect that such methods would be applicable only for small values of $N$.
\end{remark}

 \subsection{Structure and Reduced Order Characterization}
\label{ReductionSubsec}
In this section, we use the monotonicity of the optimal strategies, shown in Proposition \ref{MonotProp}, to derive a reduced order characterization of the Nash equilibrium.

The actions on a Nash equilibrium have an interesting structure. Assume that $\pi$ is a Nash equilibrium and:
\begin{equation}\label{V_def}
  \mathcal V = \{v^l\in \tilde U: \exists j: \pi_{(j-1)2^N+l}>0\}\subset \tilde U,
\end{equation}
is the set of actions used by a set of players with a positive mass. Let us define a partial ordering on $\tilde U$. For $v^1,v^2\in\tilde U$, we write  $v^1\preceq v^2$ if $v^1_k\leq v^2_k$ for all $k=1,\dots,N$. Proposition \ref{MonotProp}.(iii) implies that $\mathcal V$ is a totally ordered subset of $\tilde U$ (chain).

\begin{lemma}\label{ChainLem}
There are at most $N!$ maximal chains in $\tilde U$, each of which has  length $N+1$. Thus, at a Nash equilibrium, there are at most $N+1$ different actions in $\mathcal V$.
\end{lemma}
\begin{proof} See Appendix \ref{ProofChainLem}.
\end{proof}

For each time step $k$, denote by $\rho_k$ the fraction of players who play $u_M$, that is, $\rho_k=\mu(\{i:u^i_k=u_M\})$. Given any vector $\rho=[\rho_1\dots \rho_N]\in [0,1]^N$, we will show that  there is a unique $\pi\in\Pi$, such that the corresponding actions satisfies the conclusion of Proposition \ref{MonotProp}.(iii) and induces the fractions $\rho$.
An example of the relationship between $\pi$ and $\rho$ is given in Figure \ref{rhoFig}.

\begin{figure}[h!]
\centering
  \includegraphics[width=0.6\textwidth]{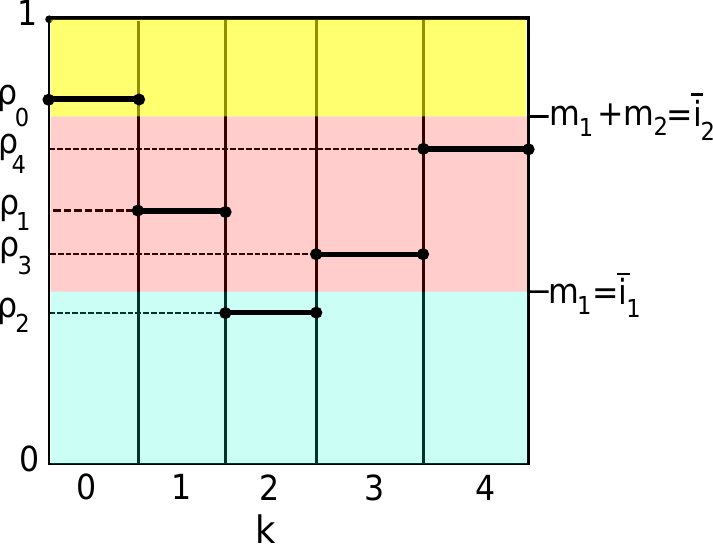}
  \caption{\textit{In this example,  $N=5$ and there are $M=3$ types of players, depicted with different colors.   The mass of players below each solid line play $u_M$ and the mass of players above the line play $u_m$. Example 1 computes $\pi$ from $\rho$.}}
  \label{rhoFig}
\end{figure}


Let us define the following sets:
$$\mathbbm I_k = \{i\in[0,1): u^i_k=u_M\},~~\mathbbm K_k = \{k': \rho_{k'}\geq \rho_k\}, ~~k=0,\dots,N-1.$$
Let  $k_1,\dots,k_N$ be a reordering of $0,\dots,N-1$ such that $\rho_{k_1}\leq \rho_{k_2}\leq \dots\leq\rho_{k_N}$. Consider also the set $\tilde {\mathcal V}=\{\bar v^1,\dots,\bar v^{N+1}\}$ of $N+1$ actions $\bar v^n$ with:
\begin{equation}
\bar v^n_k   =\begin{cases}u_M \text{~~ if~~} k\in \mathbbm K_{k_n}   \\u_m \text{~~ otherwise} \end{cases}, ~~n=1,\dots,N,\label{Actions_}
\end{equation}
and $\bar v^{N+1}_k=u_m$, for all $k$. Observe that $\bar v^{n+1}\preceq \bar v^{n}$. The following proposition shows that the set $\mathcal V$, defined in \eqref{V_def} is subset of the set $\tilde {\mathcal V}$.
\begin{proposition}
\label{Monot_rho}
Assume that $(u_k^i)_{i\in [0,1),k=0,\dots,N-1}$, with $u^i\in\tilde U$, be a set of actions satisfying the conclusions of Proposition \ref{MonotProp}. Then:
\begin{itemize}
\item[(i)] For $k\neq k'$, either $\mathbbm I_k\subset\mathbbm I_{k'}$ or  $\mathbbm I_{k'}\subset\mathbbm I_{k}$.
\item[(ii)] If for some $k,k'$, it holds $\rho_k=\rho_{k'}$ then $\mu$--almost surely all the players have the same action on $k,k'$, i.e., $\mu(\{i:u^i_k=u^i_{k'}\})=1$.
\item[(iii)] Up to subsets of measure zero, the following inclusions hold:
\begin{align*}
& \mathbbm I_{k_1}\subsetsim\mathbbm I_{k_2}\subsetsim\dots\subsetsim\mathbbm I_{k_N},\\
&\mathbbm K_{k_1}\supset\mathbbm K_{k_2}\supset\dots\supset\mathbbm K_{k_N},
\end{align*}
where $\mathbbm I_{k_n}\subsetsim\mathbbm I_{k_{n+1}}$ indicates that $\mu(\mathbbm  I_{k_{n}}\smallsetminus \mathbbm I_{k_{n+1}})=0$.
Furthermore, $\mu(\mathbbm I_{k })=\rho_{k}$.
\item [(iv)]
 For $\mu$--almost all $i\in\mathbbm I_{k_{n+1}}\smallsetminus \mathbbm I_{k_{n}}$ the action $u^i$ is given by $\bar v^{n+1}$,  for $\mu$--almost all   $i\in \mathbbm I_{k_1}$, $u^i =\bar v^1$, and for $\mu$--almost all  $i\in [0,1)\smallsetminus\mathbbm I_{k_N} $, $u^i_k =\bar v^{N+1}$.
\end{itemize}
\end{proposition}
\begin{proof} See Appendix \ref{Monot_rho_proof}.
\end{proof}

\begin{corollary}
\label{Corollary_1}
The mass of players of type $j$ with action $\bar v^n$ is given by:
 \begin{align}
 \mu(i: i \text{ is of type } j, u^i=\bar v^n) = \mu([\bar i_{j-1},\bar i_j)\cap[ \rho_{k_{n-1}},\rho_{k_{n}})),
 \end{align}
 where we use the convention that $\rho_{k_{0}}=0$, and $\rho_{k_{N+1}}=1$. Thus:
 \begin{equation}
 \label{compute_pi}
 \pi_{{(j-1)2^N+l}}=\begin{cases}  \mu([\bar i_{j-1},\bar i_j)\cap[ \rho_{k_{n-1}},\rho_{k_{n}})) \text{~~if~~} v^l = \bar v^n\\ 0~~~~~~~~~~~~~~~~~~~~~~~~~~~~~~\text{otherwise}  \end{cases}
 \end{equation}
\end{corollary}
\begin{proof} The proof follows dirrectly from of Proposition  \ref{Monot_rho} and Proposition \ref{MonotProp}.(ii).
\end{proof}
\begin{remark}
There are at most $M+N+1$ combinations of $j,l$ such that $\pi_{{(j-1)2^N+l}}>0.$
\end{remark}

Let us denote by $\tilde \pi (\rho)$ the value of vector $\pi$ computed by \eqref{compute_pi}.

\begin{example}
As an example, we compute the vector $\pi=\tilde \pi(\rho)$ for the vector of fractions $\rho$ of Figure \ref{rhoFig}.
Using Corollary \ref{Corollary_1}, we find that the possible actions are $\bar v^0 = [u_M,u_M,u_M,u_M,u_M]$, $\bar v^1 = [u_M,u_M,u_m,u_M,u_M]$,  $\bar v^2 = [u_M,u_M,u_m,u_m,u_M]$,  $\bar v^3 = [u_M,u_m,u_m,u_m,u_M]$,  $\bar v^4 = [u_M,u_m,u_m,u_m,u_m]$,  $\bar v^5 = [u_m,u_m,u_m,u_m,u_m]$. The mass of the players of each type following each action is described in the following table.  \begin{center}
 \small
\begin{tabular}{ c||c|c|c|c|c|c|c|c }
Type &1 &  1 &  2&2&2&2&3&3 \\  \hline
Mass&$\rho_2$&$\bar i_1-\rho_2$&$\rho_3-m_1$&$\rho_1-\rho_3$&$\rho_4-\rho_1$&  $\bar i_2-\rho_4$&$\rho_0-\bar i_2 $&$1-\rho_0$\\
 \hline
Action&$\bar v^0$&$\bar v^1$&$\bar v^1$&$\bar v^2$&$\bar v^3$&$\bar v^4$&$\bar v^4$&$\bar v^5$
\end{tabular}
\end{center}
\textit{Type 1 corresponds to blue, type 2 corresponds to pink and type 3 to yellow in Figure   \ref{rhoFig} }
\end{example}

\begin{proposition}
\label{Reduced}
The fractions $\rho_0,\dots,\rho_{N-1}$ correspond to a Nash equilibrium if and only if:
\begin{equation}
\label{Redu_Eq}
H(\rho)=\sum_{j=1}^M\sum_{n=1}^{N+1} \mu([\bar i_{j-1},\bar i_j)\cap[ \rho_{k_{n-1}},\rho_{k_{n}})) (\bar F_{j,\bar v^n}(\tilde\pi(\rho))-\delta^j(\tilde\pi(\rho)) )=0,
\end{equation}
where $\bar F_{j,\bar v^n}(\pi) $ is the cost of action $\bar v^n$, for a player of type $j$.
Furthermore, $H(\pi)$ is continuous and non-negative.
\end{proposition}
\begin{proof} See Appendix \ref{Reduced_proof}.
\end{proof}
\begin{remark}
The computation of an equilibrium has been reduced to the calculation of the minimum of an $N-$dimensional function. We exploit this fact in the following section to proceed with the numerical studies.
\end{remark}


\section{Numerical Examples}
\label{Sec_Num}

In this section, we give some numerical examples of Nash equilibria computation. Subsection \ref{Single_type_Num_sect} presents an example with a single type of players and \ref{Many_type_Num_sect}  an example with many types of players. Subsection  \ref{Eff_of_u_M} studies the effect the maximum allowed action $u_M$ on the strategies and the costs of the players\footnote{Data availability: The datasets generated during and analysed during the current study are available from the corresponding author on reasonable request.}. 

\subsection{Single type of players}
\label{Single_type_Num_sect}
In this subsection, we study the symmetric case, i.e., all the players have the same parameter $b^i$. The parameters for the dynamics are $r=0.4$ and $a=1/6$ which correspond to an epidemic with basic reproduction number $R_0=2.4$, where an infected person remains infectious for an average of $6$ days (these parameters are similar with \cite{kabir2020evolutionary} which analyzes COVID--19 epidemic). We assume that $u_m=0.4$ and that there is a maximum action $u_M=0.75,$ set by the government. The discretization time intervals are $1$ week and the time horizon $T$ is approximately 3 months (13 weeks). The initially infected players are $I_0=0.01$. We chose this time horizon to model a wave of the epidemic, starting at a time point where $1\%$ of the population is infected. We assume that  $\kappa=3$. 

We then compute the Nash equilibrium using a multi start local search method for \eqref{Redu_Eq}.
Figure \ref{rho_many_b} shows the fractions $\rho$, for several values of $b$, and Figure \ref{tot_infected} presents the evolution of the total mass of infected players for the same values of $b$. We observe that, for small values of $b$, which correspond to less vulnerable or very sociable agents, the players do not engage in voluntary social distancing. For intermediate values of $b$ the players engage voluntary social distancing, especially when there is a large epidemic  prevalence. For large values of $b$,  there is an initial reaction of the players which reduces the number of infected people. Then the actions of the players return to intermediate levels and keep the number of infected people moderate. In all the cases, voluntary social distancing `flattens the curve' of infected people mass.  

\begin{figure}[h!]
\centering
  \includegraphics[width=\textwidth]{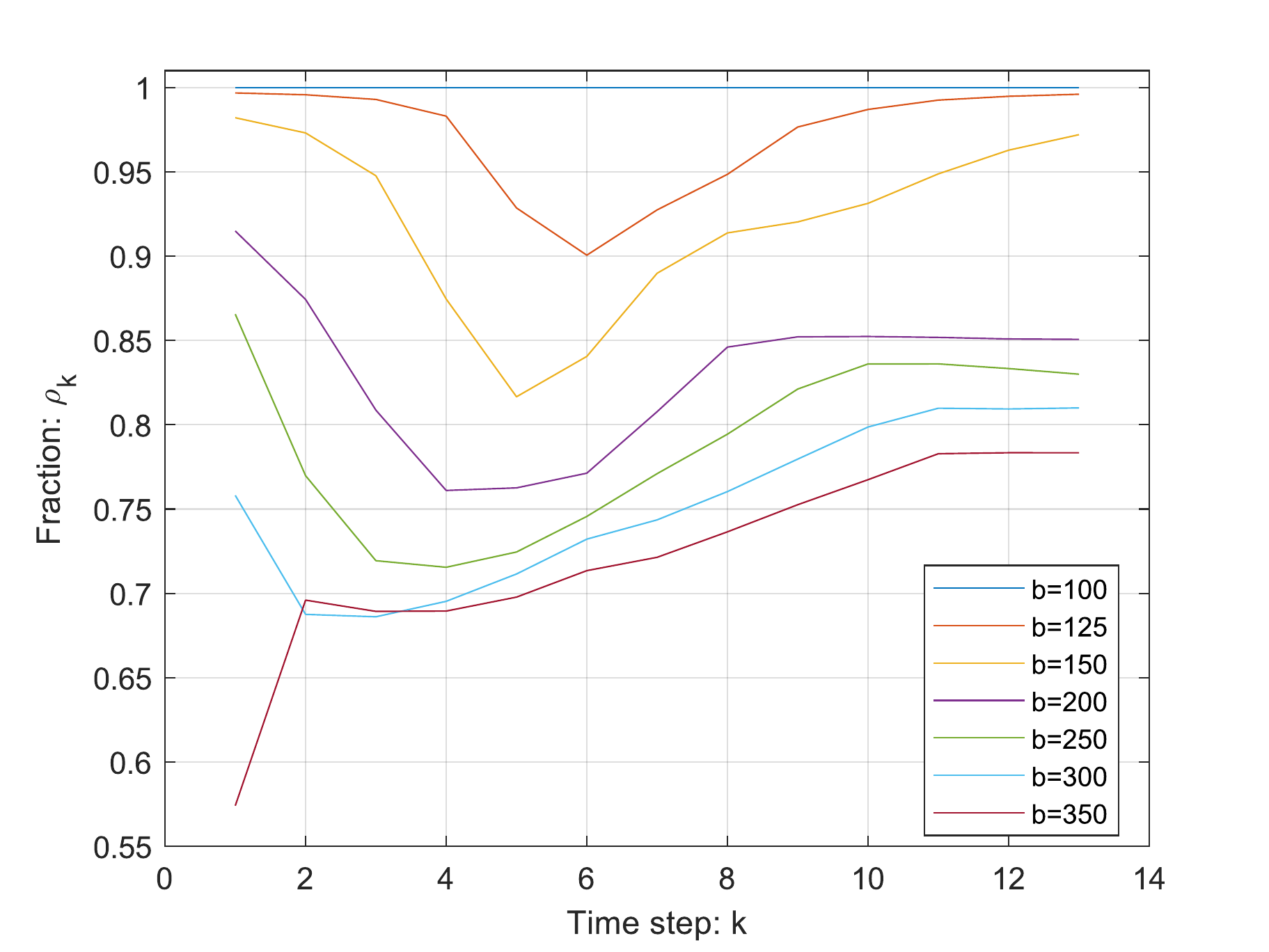}
  \caption{\textit{The fractions $\rho_k$, for $k=1,\dots,13$, for $b=100,125,150,200,250,300,350$.}}
  \label{rho_many_b}
\end{figure}

\begin{figure}[h!]
\centering
  \includegraphics[width=\textwidth]{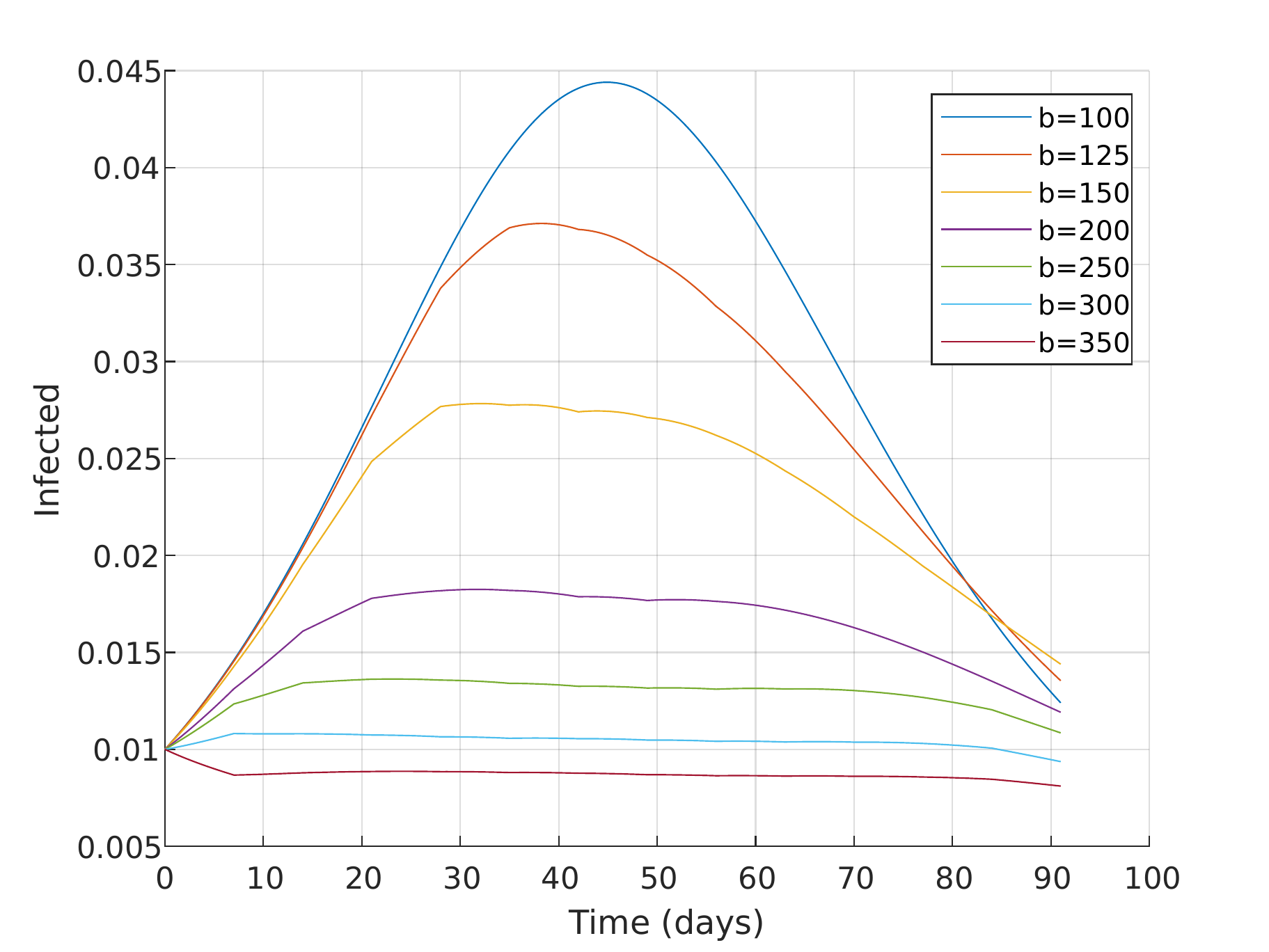}
  \caption{\textit{The evolution of the number of infected people under the computed Nash equilibrium, for $b=100,125,150,200,250,300,350$.}}
  \label{tot_infected}
\end{figure}

We then present some results for the case where $b^i=200$.  Figure \ref{SI_evol_fig} illustrates the evolution of $S^i(t)$ and $I^i(t)$, for the players belonging to different intervals $(\rho_k,\rho_{k+1})$ and thus following different strategies. We observe that the trajectories of $S^i$'s do not intersect. What is probably interesting is that the trajectories of $I^i$ may intersect. This indicates that, towards the end of the time horizon, it is probable for a person who was less cautious, i.e., she used higher values of $u^i$, to have a lower probability of being infectious.
\begin{figure}[h!]
\centering
  \includegraphics[width=\textwidth]{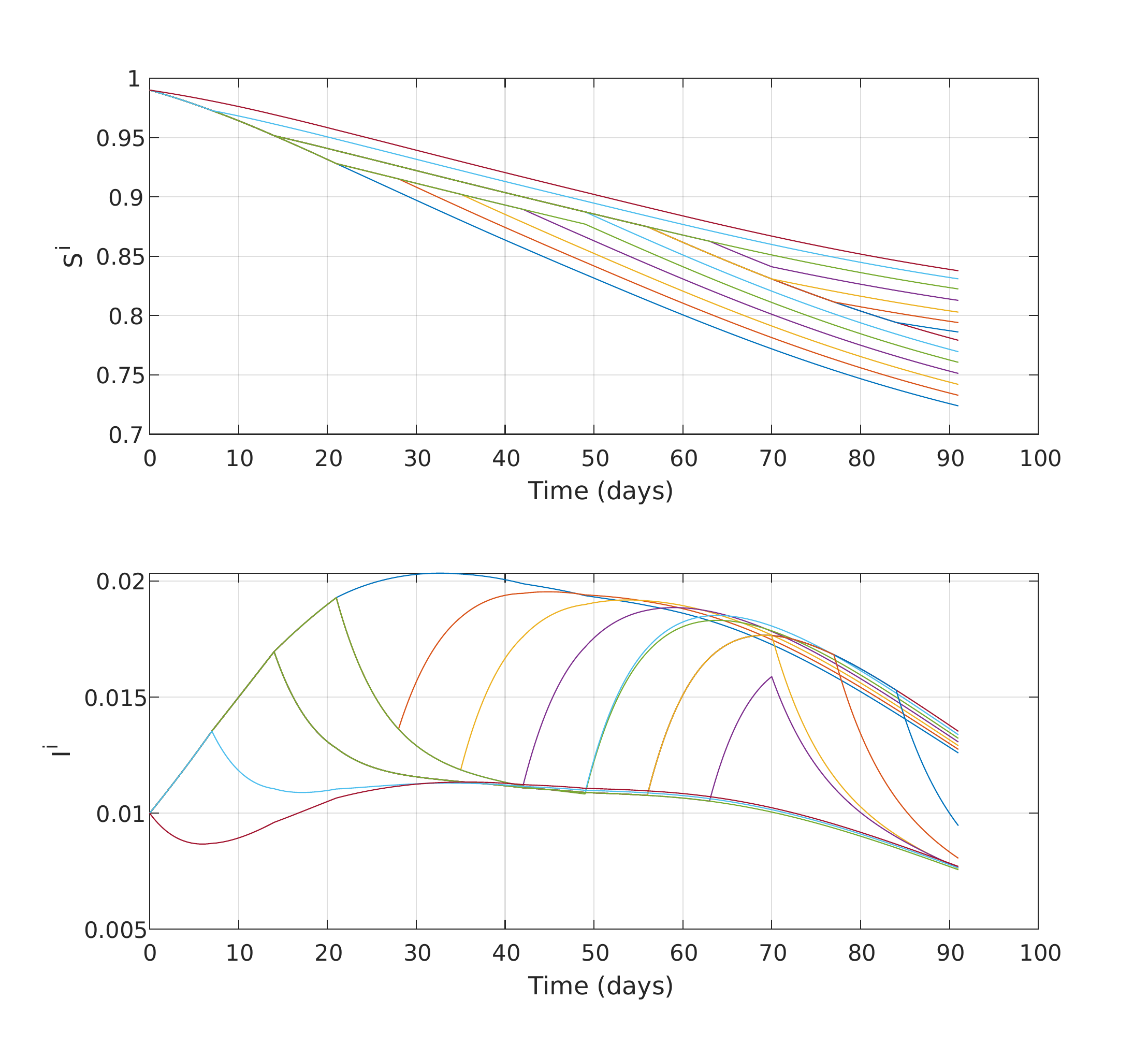}
  \caption{\textit{The evolution of the probabilities $S^i$ and $I^i$, for players following different strategies, for $b=200$.}}
  \label{SI_evol_fig}
\end{figure}

\subsection{Many Types of Players}
\label{Many_type_Num_sect}
We then compute the Nash equilibrium for the case of multiple types of players. 
We assume that there are six types of players with vulnerability parameters $G^1=100,~ G^2=200,~ G^3=400,~ G^4=800,~ G^5=1600,~  G^6=3200$. The sociability parameter $s^i$ is equal to $1$, for all the players.  The masses of these types are $m_1=0.5$ and $m_2=\dots=m_6=0.1$.  The initial condition is for all the players $I_0=0.0001$ and the time horizon is 52 weeks (approximately a year). Here we assume that the maximum action is $u_M=0.8$. The rest of the parameters are as in the previous subsection. 

Figure \ref{fractions_many_fig} shows the fractions $\rho$ and Figure \ref{many_types_dynamics} presents the evolution of the probability of each category of players to be susceptible and infected. Let us note that the analysis of Subsection  \ref{ReductionSubsec} simplifies a lot the analysis. Particularly, the set $\Pi$ has $(2^{52}-1)^6\simeq 8.3\cdot 10^{93}$ dimensions, while Problem \eqref{Redu_Eq} only 52. 

\begin{figure}[h!]
\centering
  \includegraphics[width=\textwidth]{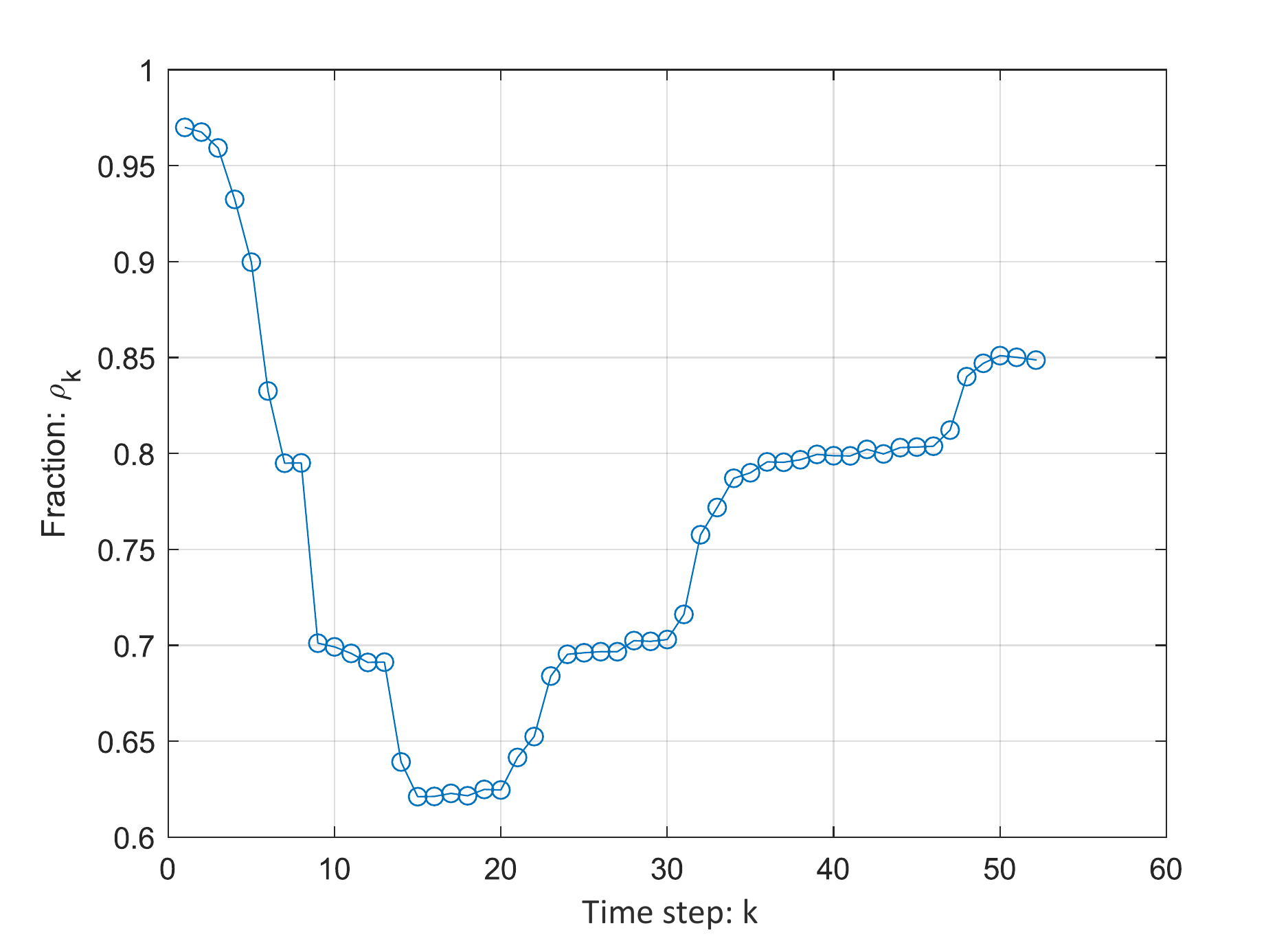}
  \caption{\textit{The fractions $\rho_k$}}
  \label{fractions_many_fig}
\end{figure}
\begin{figure}[h!]
\centering
  \includegraphics[width=\textwidth]{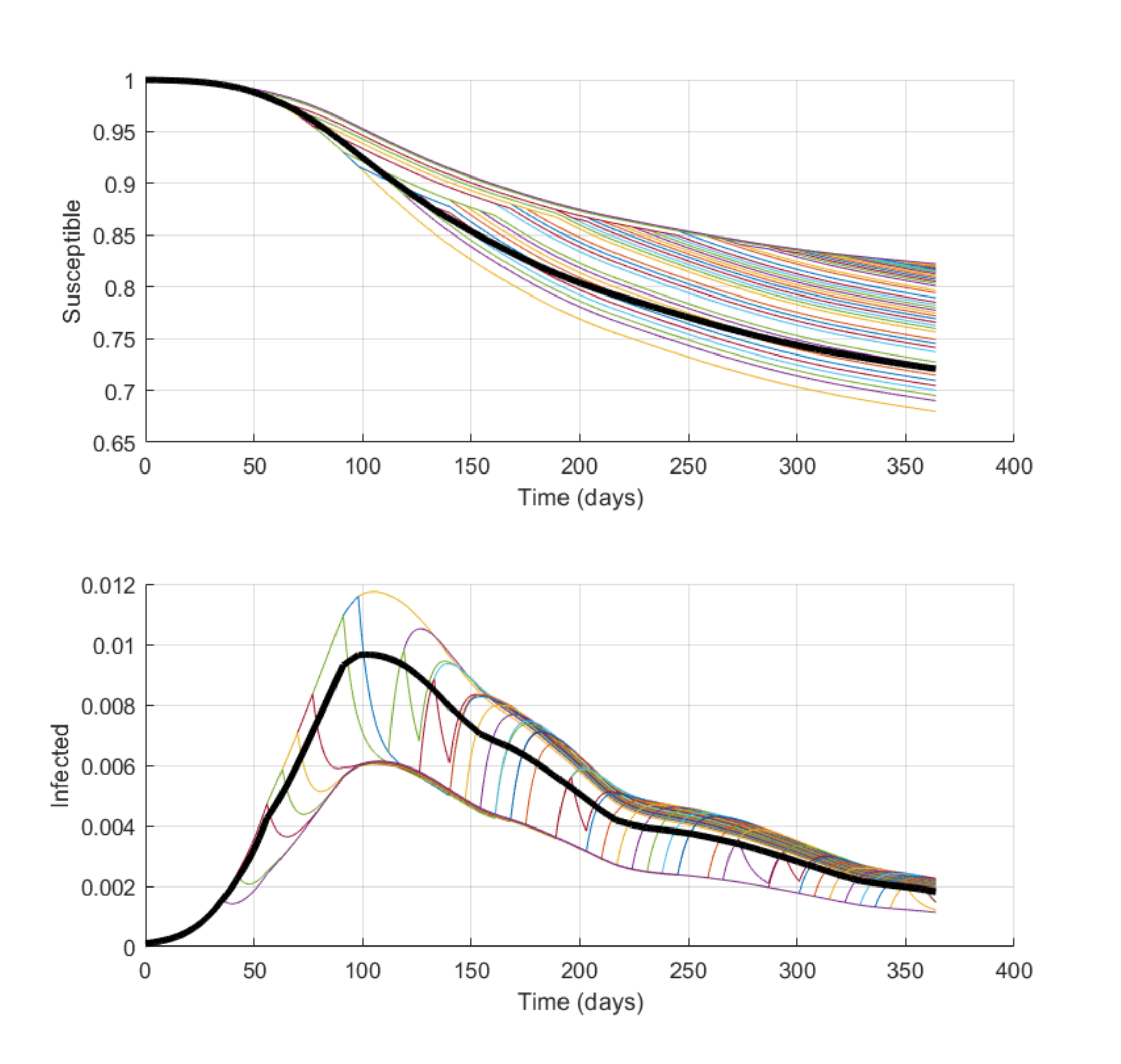}
  \caption{\textit{The probability of a class of people to be susceptible and infected. The colored lines correspond to the probabilities of being susceptible $S^i(t)$ and infected $I^i(t)$, for the several classes of players. The bold black line represents the mass of susceptible and infected persons}}
  \label{many_types_dynamics}
\end{figure}

\subsection{Effect of $u_M$}
\label{Eff_of_u_M}

We then analyze the case where the types of the players are as in subsection \ref{Many_type_Num_sect}, the initial condition is $I_0=0.005$ for all the players, for various values of $u_M$. The time horizon is 13 weeks.

Figure \ref{rho_many_u_M} illustrates the equilibrium fractions $\rho_k$, for the various values of $u_M$. We observe that as $u_M$ increases, the fractions $\rho_k$ decrease. Figure \ref{many_types_INFECTED} shows the evolution of the mass of infected players, for the different values of the maximum action $u_M$. We observe that,  as $u_M$ increases, the mass of infected players decreases. Figure \ref{many_types_costs} presents the cost of the several types of players, for the different value of the maximum action $u_M$. We observe that players with low vulnerability ($G=100$) prefer always a larger value of $u_M$, which corresponds to less stringent restrictions. For vulnerable players (e.g. $G=3200$) the cost is an increasing function of $u_M$, that is they prefer more stringent restrictions. For intermediate values of $G$, the players prefer intermediate values of $u_M$. The mean cost in this example is minimized for $u_M=0.6$.

\begin{figure}[h!]
\centering
  \includegraphics[width=\textwidth]{ 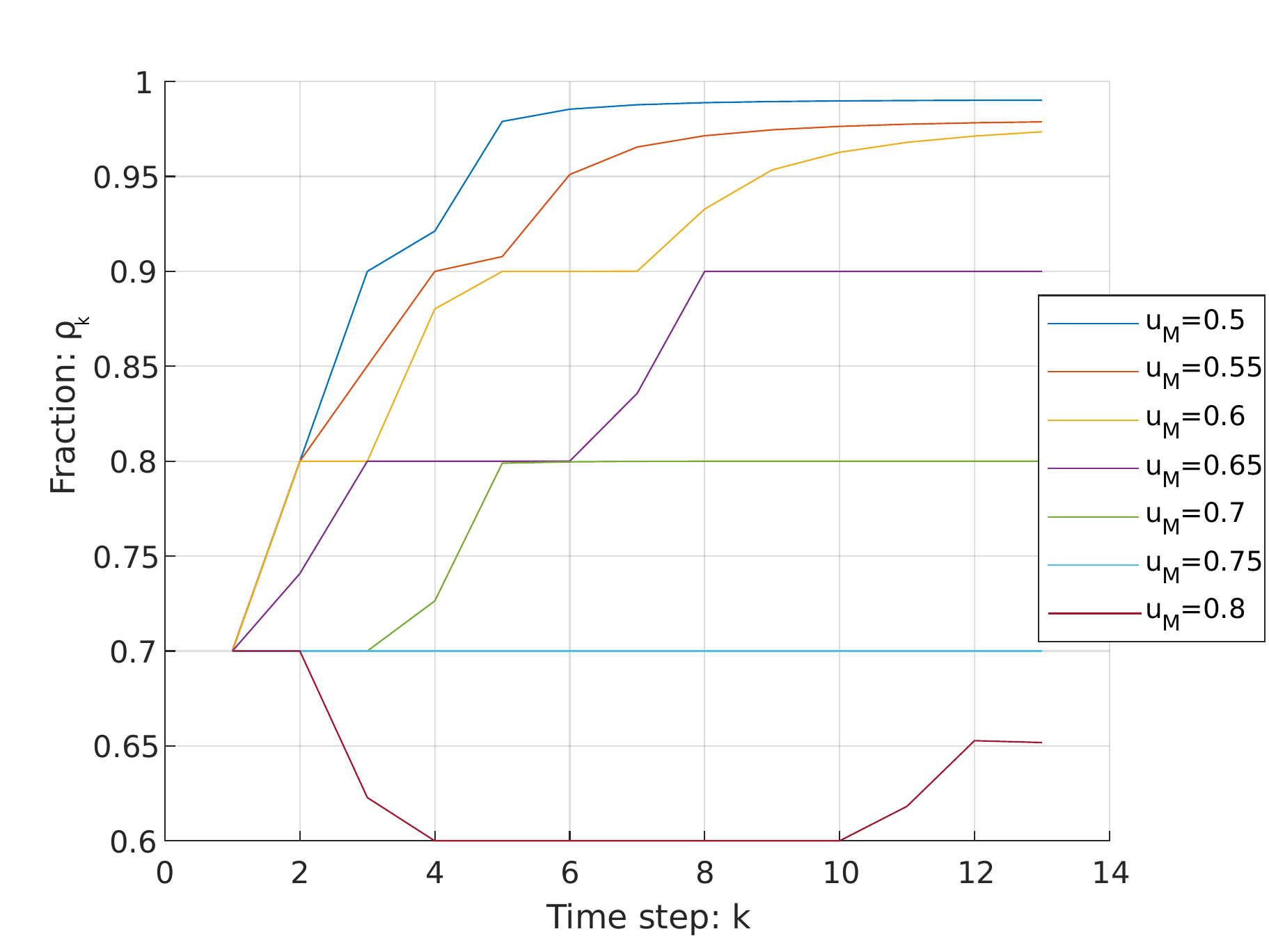}
  \caption{\textit{ The fractions $\rho_k$, for the several values of the maximum action $u_M$.}}
  \label{rho_many_u_M}
\end{figure}

\begin{figure}[h!]
\centering
  \includegraphics[width=\textwidth]{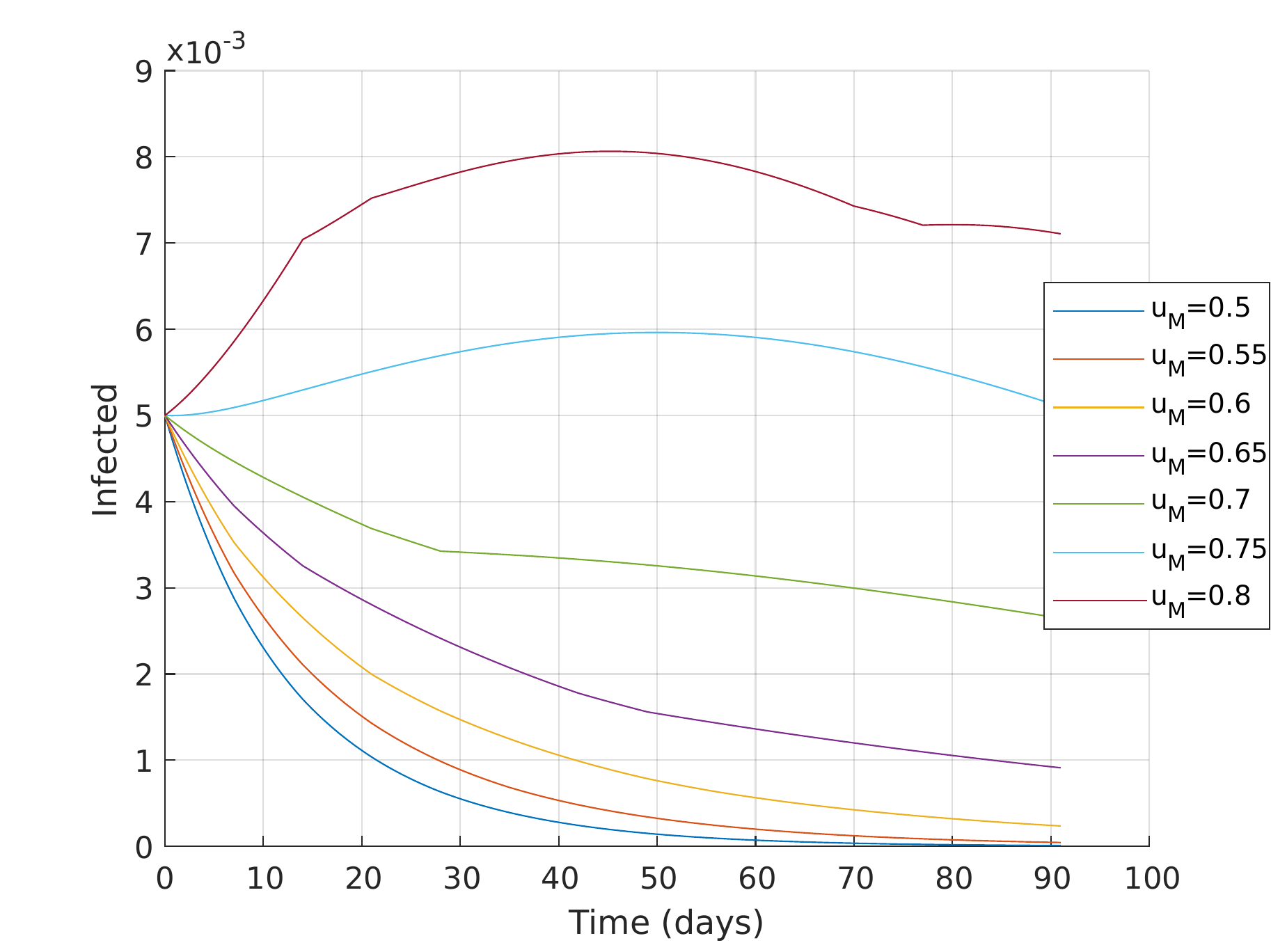}
  \caption{\textit{  The mass of infected people as a function of time, for the different values of the maximum action $u_M$.}}
  \label{many_types_INFECTED}
\end{figure}
 
\begin{figure}[h!]
\centering
  \includegraphics[width=\textwidth]{ 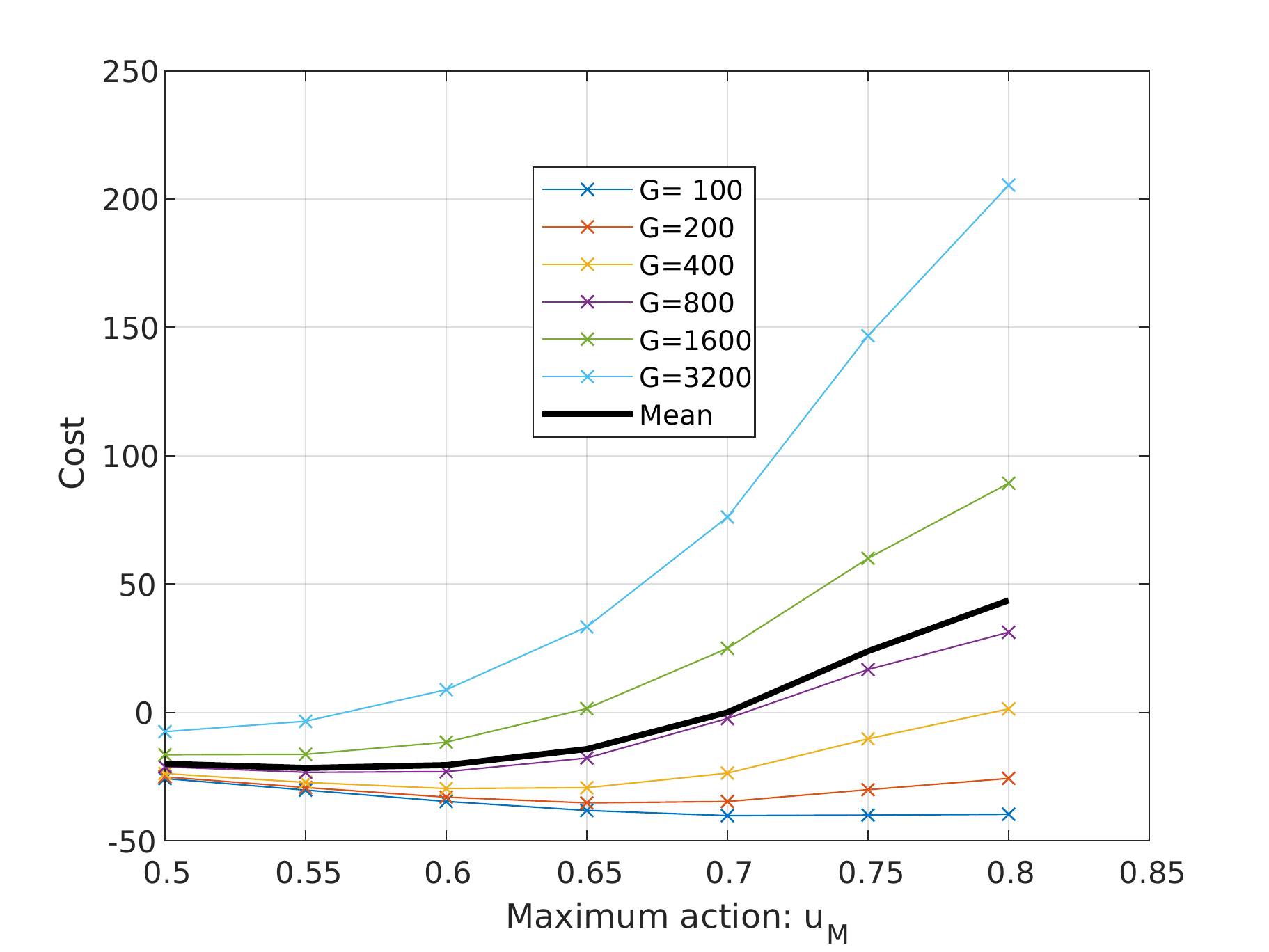}
  \caption{\textit{ The cost for the several classes of players, for the different values of the maximum action $u_M$. The bold black line represents the mean cost of all the players. }}
  \label{many_types_costs}
\end{figure}

\section{Conclusion}

This paper studied a dynamic game of social distancing during an epidemic, giving an emphasis on the analysis of asymmetric solutions. We proved the existence of a Nash equilibrium and derived some monotonicity properties of the agents' strategies. The monotonicity result was then used to derive a reduced--order characterization of the Nash equilibrium, simplifying its computation significantly. Through numerical experiments, we show that both the agents' strategies and the evolution of the epidemic depend strongly on the agents' parameters (vulnerability, sociality) and the epidemic's initial spread. Furthermore, we observed that agents with the same parameters could have different behaviors, leading to rich, high--dimensional dynamics. We also observe that more stringent constraints on the maximum action (set by the government) benefit the more vulnerable players at the expense of the less vulnerable. Furthermore, there is a certain value for the maximum action constant that minimizes the average cost of the players.

There are several directions for future work. First, we can study more general epidemics models than the SIR. Second, we can investigate different information patterns, including the cases where the agents receive regular or random information about their health state. Finally, we can compare the behaviors computed analytically with real-world data.

\appendix
\section{Appendix: Proof of the Results of the Main Text}
\subsection{Existence of Solution to \eqref{individual_dynamics}}
\label{ProofExistenceDE}

For any $i\in[0,1)$, if   $[S^i(t),~I^i(t)]$ solve the differential equations \eqref{individual_dynamics}, with initial condition $(S^i(0),I^i(0))\in[0,1]^2$, then $(S^i(t),I^i(t))$  remain in $[0,1]^2$. Thus, we consider the solution of the differential equations:
\begin{equation}\label{individual_dynamics_sat}
\begin{aligned}
\dot {S}^i &= \text{sat}_B(-ru^iS^iI^f(t))\\
\dot {I}^i &=\text{sat}_B(ru^iS^iI^f(t)-\alpha^i I^i)
\end{aligned},
\end{equation}
where $\text{sat}_B(z) = \max( \min(z,B),-B)$, and $B=ru_M^2+\max_j\alpha_j$.

 Consider the Banach space $X=L^1([0,1),\mathbb R^2)$, and let $x_0=(S^\cdot(0),I^\cdot(0)):[0,1)\rightarrow \mathbb R^2$. Then, under Assumptions 1,3,  it holds $x_0\in X$. For each interval $[t_k,t_{k+1})$, the differential equations \eqref{individual_dynamics_sat} with the corresponding initial conditions  can be written as:
\begin{equation}
\label{individual_dynamics_sat_compact}
\dot x=f_k(x),~~x(t_k)=x^k_0,
\end{equation}
 where for $x:i\mapsto [S^i,~I^i]^T$, the value of  $f_k(x)\in X$ is given by: $$f_k(x):i\mapsto [\text{sat}_B(-ru^i_k S^i\mathcal M_k x), ~\text{sat}_B(ru^i_k S^i\mathcal M_k x-\alpha^i I^i)]^T,$$ where $\mathcal M_k :X\rightarrow \mathbb R$ is a linear bounded operator with $\mathcal M_k x = \int I^i u_k^i \mu(di)$. For the  initial condition it holds $x_0^0=x_0$, and $x_0^k=x(t_k)$  is computed from the solution of \eqref{individual_dynamics_sat_compact} on the interval $[t_{k-1},t_k)$, for $k\geq 1$. 
 For all $k$, $f_k$ is Lipschitz and thus there is a unique solution to \eqref{individual_dynamics_sat_compact} (e.g. Theorem 7.3 of \cite{Brezis}).
 Furthermore, both $I^\cdot(t)$ and $u^\cdot(t)$ are measurable and bounded. Thus, the integrals in  \eqref{u_tilde}, \eqref{I_F_def} are well-defined.

Note that from Assumption 1, we only used the fact that $S^\cdot(0),I^\cdot(0):[0,1)\rightarrow \mathbb R$ are measurable and not the piecewise constant property.

\subsection{Proof of Proposition \ref{EquProp}}
\label{ProofEquProp}
 (i)
  Since, $rI^f_k>0$ and $b^i>0$,   the cost \eqref{AuxilCost} is strictly concave, with respect to $u_k^i$. Thus, the minimum with respect to $u_k^i$ is either $u_m$ or $u_M$.
  \newline
  (ii) Since $U$ is compact and $\tilde J$ is continuous, there is an optimal solution for \eqref{aux_opt_pr}. Denote by $u^{i,\star}$ this solution. Further, denote by $V_1=\tilde J^i(u^{i,\star})$, and $V_2 = \inf_A \{-b^ie^{-A}+f(A)\}$  the values of problems \eqref{aux_opt_pr} and \eqref{Probl2} respectively. Then, for $A^\star=\sum_{k=0}^{N-1}ru_k^{i,\star}I_k^f$, we have $$V_2\leq -b^ie^{-A^\star}+f(A^\star)\leq -b^i\exp\left[-r\sum_{k=0}^{N-1}u^{i,\star}_kI^f_k\right]-\sum_{k=0}^{N-1}u^{i,\star}_k\bar{u}_k=\tilde J^i(u^{i,\star})=V_1,$$
where the first inequality is due to the definition of $V_2$ and the second inequality is due to the definition of $f(A^\star)$. To contradict assume that $V_2<V_1$. Then, there is some $A$ and an $\varepsilon>0$ such that:
\begin{align}
\label{TechIneq}
-b^ie^{-A}+f(A)<V_1-2\varepsilon.	
\end{align}
Thus, there is a $\tilde u^i$ such that $A=\sum_{k=0}^{N-1}r\tilde u_k^{i}I_k^f$ and $\sum_{k=0}^{N-1}\tilde u^{i}_k\bar{u}_k<f(A)+\varepsilon$. Combining with \eqref{TechIneq} we get
$ \tilde J^i(\tilde u)<V_1-\varepsilon	$, which contradicts the definition of $V_1$.

For $u^{i,\star}$ minimizing \eqref{aux_opt_pr}, the problem \eqref{Probl2} is minimized for $A^\star=\sum_{k=0}^{N-1}ru_k^{i,\star}I_k^f$ and $u^{i,\star}$ attains the minimum in \eqref{f_def}. To see this, observe that otherwise we would have $V_2<V_1$. On the other hand, assume that $A$ minimizes \eqref{Probl2} and $ u^i$ attains the minimum in \eqref{f_def}. Then, $V_2=-be^{-A}+f(A)=J^i(u^i)$. Furthermore, since $V_2=V_1$, it holds $J^i(u^i)=V_1$, and thus $u^i$ minimizes $J^i$.
\newline
(iii) The set $\left\{ u^i\in U:  \sum_{k=0}^{N-1}u^i_kI^f_k=A/r  \right\}$ is nonempty if and only if $A\in[A_m,A_M]$. Thus, the $f(A)$ is finite if and only if $A\in[A_m,A_M]$.

For $A\in[A_m,A_M]$, there exists an optimal solution $u^i$ that attains the minimum in \eqref{f_def}. Since, \eqref{f_def} is a feasible linear programming problem, there is a Lagrange multiplier $\lambda$ (e.g. Proposition 5.2.1 of \cite{Bert}), and $u^i$ minimizes the Lagrangian:
 \begin{align}
 L(u^i,\lambda)=-\sum_{k=0}^{N-1}\bar u_k u_k^i +\lambda\sum_{k=0}^{N-1}I_k^fu_k^i-\lambda A/r.
\label{Lagrangian}
 \end{align}
Thus, $u_k^i=u_m$, if $\bar u_k/I^f_k<\lambda$ and $u_k^i=u_M$, if $\bar u_k/I^f_k>\lambda$. To compute $f(A)$, we reorder $k$, using a new index $k'$, such that $\bar u_{k'}/I^f_{k'}$ is non-increasing.
Let:
$$k'_A=\	max\{\bar k_A': \sum_{k'=0}^{\bar k_A'-1} u_MI^f_{k'}+  \sum_{k'=\bar k_A'}^{N-1} u_mI^f_{k'}\leq A/r\}.$$
Then:
$$\Sigma_{k_A'}+ u_{\bar k_A'}^iI^f_{\bar k_A'} = A/r.$$
where $\Sigma_{k_A'}=\sum_{k'=0}^{\bar k_A'-1} u_MI^f_{k'}+  \sum_{k'=\bar k_A'+1}^{N-1} u_mI^f_{k'}$.
Thus:
$$f(A)=-\sum_{k'=0}^{k_A'-1} \bar u_{k'} u_M-\sum_{k'=k_A'+1}^{N-1}  \bar u_{k'} u_m-\frac{\bar u_{k_A'}}{I^f_{k_A'}}(A/r -\Sigma_{k_A'}),$$
for $A/r\in[\Sigma_{k_A'}+ u_m I^f_{ k_A' } , \Sigma_{k_A'}+ u_M I^f_{ k_A' } ]$. It holds:$$\Sigma_{k_A'}+ u_M I^f_{ k_A' }=\Sigma_{k_A'+1}+ u_m I^f_{k_A' +1}.$$  Therefore, $f$ is continuous and  piecewise affine. Furthermore, since $\bar u_{k'}/I^f_{k'}$ is non-increasing with respect to $k'$, the slope of $f$ is non-decreasing, i.e., it is convex. Thus, $f$ is differentiable in all $(A_m,A_M)$ except of the points $A=\Sigma_{k'}+u_MI_{k'}^f$ with $\bar u_{k'}/I^f_{k'}>\bar u_{k'+1}/I^f_{k'+1}$.  The linear parts of $f$ are at most $N$.
\newline
(iv) Since $-b^ie^{-A}$ is strictly concave in $A$, there are at most $N+1$ possible minima of $ -b^ie^{-A}+f(A)$, which correspond to the points of non-differentiability of $f$ in $(A_m,A_M)$ and the points $A_m$ and $A_M$. Observe that for $A=A_m$ or $A=A_M$, there is a unique $u^i$ minimizing \eqref{f_def}.

We then show that for all the non-differentiability points $A$ of $f$  there is a unique $u^i$ minimizing \eqref{f_def}. If $A$ is a non-differentiability point, there is a $k'_0$ such that $A/r=\sum_{k'=0}^{k'_0}I_{k'}^fu_M+\sum_{k'=k'_0+1}^{N-1}I_{k'}^fu_m$ and $\bar u_{k'_0}/I^f_{k'_0}>\bar u_{k'_0+1}/I^f_{k'_0+1} $. We then show that the the unique minimizer in \eqref{f_def} is given by $u^i_{k'}=u_M$ for $k'\leq k'_0$ and  $u^i_{k'}=u_m$ for $k'> k'_0$. Indeed, $u^i$ is feasible and if $u'\neq u^i $ is another feasible point it holds:
$$ \sum_{k'=0}^{k'_0} (u_M-u'_{k'})I^f_{k'} +\sum_{k'=k'_0+1}^{N-1}(u_m-u'_{k'})I^f_{k'}=0. $$
Multiplying by $\bar u_{k'_0}/I^f_{k'_0}$ we get:
$$ \sum_{k'=0}^{k'_0} (u_M-u'_{k'})\frac{I^f_{k'}\bar u_{k'_0}}{I^f_{k'_0}} +\sum_{k'=k'_0+1}^{N-1}(u_m-u'_{k'})\frac{I^f_{k'}\bar u_{k'_0}}{I^f_{k'_0}}=0. $$
Then, using that $u_M-u'_{k'}\geq 0$, $u_m-u'_{k'}\leq 0$, and that for $k'\leq k'_0$ it holds $\bar u_{k'}/I^f_{k'}\geq \bar u_{k'_0}/I^f_{k'_0}$ and for $k'> k'_0$ it holds $\bar u_{k'}/I^f_{k'}<\bar  u_{k'_0}/I^f_{k'_0}$ we have:
$$ -\sum_{k'=0}^{N-1}u'_{k'}\bar u_{k'} -\left[-\sum_{k'=0}^{N-1}u^i_{k'}\bar u_{k'}\right] = \sum_{k'=0}^{k'_0} (u_M-u'_{k'})\frac{I^f_{k'}\bar u_{k'}}{I^f_{k'}} +\sum_{k'=k'_0+1}^{N-1}(u_m-u'_{k'})\frac{I^f_{k'}\bar u_{k'}}{I^f_{k'}}\geq 0, $$
and the inequality is strict if for some $k'>k'_0$, $u'_{k'}\neq u_m$. Therefore, $u^i$ is optimal and if $u'$ is also optimal then it should satisfy $u'_{k'} = u_m$ for all $k'>k'_0$. Combining this with the fact that $\sum_{k'=0}^{N-1} u'_{k'} I^f_{k'}=A/r$ and $I^f_{k'}>0$, we get $u'=u^i$.
\newline
(v) We have shown that if $u^i$ is optimal then there is a $k'_0$ such that $u^i_{k'}=u_M$ for $k'\leq k'_0$ and $u^i_{k'}=u_m$ for $k'>k'_0$. Then,  using the original index $k$, the optimal control can be expressed as:
$$u^i_k = \begin{cases} u_M ~~~\text{if ~}  \bar u_k/I_k^f\geq \lambda'\\
 u_m ~~~\text{if ~}  \bar u_k/I_k^f< \lambda' \end{cases},$$
where $\lambda' = \bar u_{k_0'}/I^f_{k_0'}$.

\subsection{Proof of Proposition \ref{MonotProp}}
\label{ProofMonotProp}
(i)
 Since $A_1$ is optimal for $b^{i_1}$ and $A_2$ is optimal for $b^{i_2}$ it holds:
\begin{align*}
-b^{i_1}e^{-A_1}+f(A_1)\leq -b^{i_1}e^{-A_2}+f(A_2),\\
-b^{i_2}e^{-A_2}+f(A_2)\leq -b^{i_2}e^{-A_1}+f(A_1).
\end{align*}
Adding these equations and reordering, we get:
\begin{align*}
(b^{i_2}-b^{i_1})e^{-A_2}\geq (b^{i_2}-b^{i_1})e^{-A_1}.
\end{align*}
And since $b^{i_2}>b^{i_1}$, we get $A_1\geq A_2$
 \newline
 (ii) Using (v) of Proposition \ref{EquProp}, and $A_1\geq A_2$ we get:
 \begin{align*}
 A_1/r = \sum_{k=0}^{N-1} u_k^{i_1}I^f_k = \sum_{k=0}^{N-1} \left(u_m+(u_M-u_m)h_{\lambda'_1}(\bar u_k/I^f_k)\right)I^f_k \geq \\~~~~~~~\geq \sum_{k=0}^{N-1} \left(u_m+(u_M-u_m)h_{\lambda'_2}(\bar u_k/I^f_k)\right)I^f_k = \sum_{k=0}^{N-1} u_k^{i_2}I^f_k= A_2/r
 \end{align*}
 where $h_\lambda(\cdot) $  is the Heaviside function, i.e., $h_{\lambda'}(x)=1 $ if $x\geq \lambda'$ and $h_{\lambda'}(x)=0$ otherwise. Therefore, $\lambda'_1\leq \lambda'_2$.
 \newline
(iii)
Assume that for $k_1\neq k_2$, $u^{i_1}_{k_1}=u^{i_2}_{k_2}=u_m$ and $u^{i_2}_{k_1}=u^{i_1}_{k_2}=u_M$. Then, using (v) of Proposition \ref{EquProp} we have:
 \begin{align*}
 \lambda'_2< \frac{\bar u_{k_1}}{I^f_{k_1}}\leq \lambda'_1,~~~
 \lambda'_1\leq \frac{\bar u_{k_2}}{I^f_{k_2}}< \lambda'_2,
 \end{align*}
which is a contradiction.

\subsection{Proof of Proposition \ref{NCP_Prop}}\label{ProofNCP_Prop}
(i) Assume that a $\pi\in\Pi$ satisfies \eqref{NCP_large} and fix a $j\in\{1,\dots,M\}$. For any $l$ such that $\pi_{(j-1)2^N+l}>0$, it holds $\Phi_{(j-1)2^N+l}(\pi)=0$, that is $F_{(j-1)2^N+l}(\pi)=\delta^j(\pi)=\min_{ l'} F_{(j-1)2^N+l'}(\pi)$. Thus,  $\pi\in \Pi$ is a Nash equilibrium.

Conversely, assume that $\pi\in \Pi$ is a Nash equilibrium and fix a $j\in\{1,\dots,M\}$. There is an $l$ such that $\pi_{(j-1)2^N+l}>0$. Since $\pi$ is a Nash equilibrium,
it holds $F_{(j-1)2^N+l}(\pi)=\delta^j(\pi)$ and for all other $l'$ it holds  $F_{(j-1)2^N+l'}(\pi)\geq F_{(j-1)2^N+l}(\pi)=\delta^j(\pi)$, which implies \eqref{NCP_large}.
\newline
(ii) Assume that $\pi$ is a Nash equilibrium and $\pi'\in\Pi$. Then,  $\sum_{l=1}^{2^N}\pi_{(j-1)2^N+l}=\sum_{l=1}^{2^N}\pi'_{(j-1)2^N+l}=m_j$.  Since $\pi$ is a Nash equilibrium  it holds: $$\sum_{l=1}^{2^N} (\pi'_{(j-1)2^N+l}-\pi_{(j-1)2^N+l})^TF_{(j'-1)2^N+l}(\pi)\geq0.$$  Thus, \eqref{VI_large} holds.

Conversely, assume that \eqref{VI_large} holds, for some $\pi \in \Pi$. If $\pi$ is not a Nash equilibrium, then there is a $j,l$ such that $\pi_{(j-1)2^N+l}>0$ and $F_{(j-1)2^N+l}>\delta^j(\pi)$. Then, if $l'$ is such that $F_{(j-1)2^N+l'}=\delta^j(\pi)$. Taking $\pi'=\pi+\pi_{(j-1)2^N+l}e_{{(j-1)2^N+l'}}-\pi_{(j-1)2^N+l}e_{{(j-1)2^N+l}}$, we get $(\pi'-\pi)^TF(\pi)<0$, which is a contradiction.
\newline
(iii) With a slight abuse of notation we write $I^f_k(\pi), \bar u_k(\pi)$ to describe the quantities $I^f_k, \bar u_k$ when the distribution of actions is $\pi$  and $\tilde J_j(v^l,\pi)$ to describe the auxiliary cost of a player of type $j$ who plays action $v^j$ when the distribution of the actions is $\pi$.

  \begin{lemma}\label{l1}
The quantities  $I^f_k(\pi), \bar u_k(\pi), \tilde J_j({v^l},\pi)$, are continuous  on $\pi$.
\end{lemma}
\begin{proof}
The state of the system evolves according to the set of $M2^{N+1}+1$ differential equations:
\begin{align*}
\dot {S}^{j,v^l} &= -r v^l_k S^{j,v^l}I^f,\\
\dot {I}^{j,v^l} &=r v^l_k S^{j,v^l}I^f-\alpha_j I^{j,v^l},\\
\dot z &= I^f,
\end{align*}
where $j=1,\dots,M$, $l=1,\dots,2^N$, $k: t\in [t_k,t_{k+1})$, and: $$I^f  = \sum_{j=1}^M\sum_{l=1}^{2^N} \pi_{(j-1)2^N+l }{I}^{j,v^l}v^l_k. $$ The initial conditions are ${S}^{j,v^l}(0)={S}^{j}(0)$, ${I}^{j,v^l}(0)={I}^{j}(0) $, (Assumption 1) and $z(0)=0$.

The right-hand side of the differential equations depend continuously on $\pi$ through the term $I^f$.  Furthermore, $S^{j,v^l}(t),I^{j,v^l}(t)$ remain in $[0,1]$ for all $j,v^l$. Thus, the state space of the system remains in $[0,1]^{M\cdot 2^N}\times \mathbb R$
 the right-hand side of the differential equation is Lipschitz. Thus, Theorem 3.4 of \cite{Khalil} applies and ${S}^{j,v^l}(t)$,  ${I}^{j,v^j}(t)$ and $z(t)$  depend  continuously on $\pi$.
Thus, $I^f_k = z({t_k+1}) -z({t_k})$ depends continuously on $\pi$. Furthermore, $\bar u_k$ is continuous (linear) on $\pi$. Finally, the auxiliary cost $J_j(v^l,\pi)$, due to its form \eqref{AuxilCost}, depends continuously   $\pi$.
\end{proof}

To complete the proof observe that $F(\pi)$ is continuous  and $ \Pi$ is compact and convex. Thus, the existence  is a consequence of Corollary 2.2.5 of \cite{FaciPang}.

\begin{remark}
An alternative would be to use Theorem 1 of \cite{Schmeidler} or Theorem 1 of \cite{Mas_Colell}, combined with Lemma \ref{l1} to prove the existence of a mixed Nash equilibrium and then use Assumption 1, to construct a pure strategy equilibrium. However, the reduction to an NCP is useful computationally.
 \end{remark}

\subsection{Proof of Lemma \ref{ChainLem}}
\label{ProofChainLem}

 Every maximal chain begins with the least element $[u_m,\dots,u_m]$ and ends at the greatest element $[u_M,\dots,u_M]$.
  Every two consecutive elements of a maximal chain $v^l$, $v^{l+1}$ differ at exactly one point, otherwise there exists a vector $v'$: $v^l\preceq v'\preceq v^{l+1}$ and thus the chain is not maximal.

  Thus, beginning from $[u_m,\dots,u_m]$ and changing at each step one point from $u_m$ to $u_M$ we get a sequence of $N+1$ ordered vectors. So, every maximal chain has length $N+1$.

  Then, we prove that the number of such chains is $N!$ using induction.

  For $N=2$ it is easy to verify that we have two chains of 3 elements.

  Let for $N=n$ we have $n!$ maximal chains of $n+1$ elements.
  Then for $N=n+1$ we consider one of the previous chains $v^1\preceq v^2\preceq\dots\preceq v^N$ and at each of its elements we add an extra bit: $\tilde{v}^{i}=[v^i,\beta^i]$. We observe that if $\beta^i=u_M$ then for all $j>i$ it should hold $\beta^j=u_M$, in order for the new vectors to remain ordered under $\preceq$.

  Denote by $i_c$ the point that $\beta^i$ change from $u_m$ to $u_M$. For each choice of $i_c: \quad \beta^j=u_m, \quad j<i_c \quad \textrm{and} \quad \beta^j=u_M, \quad j>i_c$ we take two ordered vectors $\tilde{v}^{i_c}_1=[v^{i_c},u_m]$ and $\tilde{v}^{i_c}_2=[v^{i_c},u_M]$ in the new chain, so we have two $\beta^{i_c}$. Thus, we have $N+1$ possible choices for the $\beta=[\beta^i]\in \{u_m,u_M\}^{N+1}$. This way we observe that from each chain in $(\{u_m,u_M\}^N,\preceq)$ we can construct $N+1$ chains in $(\{u_m,u_M\}^{N+1},\preceq)$.

\begin{remark}
The fact that $\mathcal V$ has at most $N+1$ elements is also a consequence of Corollary \ref{StrategyCorollary}.
\end{remark}

\subsection{Proof of Proposition \ref{Monot_rho}}
\label{Monot_rho_proof}

(i) To contradict assume that $\mathbbm I_k\not\subset\mathbbm I_{k'}$ and  $\mathbbm I_{k'}\not\subset\mathbbm I_{k}$. Then, there is a pair of players $i_1,i_2$ such that $u^{i_1}_k=u^{i_2}_{k'}=u_M$ and $u^{i_2}_k=u^{i_1}_{k'}=u_m$, which contradicts Proposition \ref{MonotProp}.(iii).

(ii) Without loss of generality assume that $\mathbbm I_k\subset \mathbbm I_{k'}$. Then: $$\mu(\mathbbm I_k)=\rho_k = \rho_{k'} = \mu(\mathbbm I_k')=\mu(\mathbbm I_k)+\mu(\mathbbm I_{k'}\smallsetminus\mathbbm I_{k'}).$$
 Thus, $\mu(\mathbbm I_{k'}\smallsetminus\mathbbm I_{k})=0$. Combining with  $\mathbbm I_k\subset \mathbbm I_{k'}$ and the definition of $\mathbbm I_k, \mathbbm I_{k'}$ we get $\mu(\{i:u^i_k=u^i_{k'}\})=1$. 

(iii)  The equality $\mu(\mathbbm I_{k })=\rho_{k}$, is immediate from the definition of $\rho_k$. Consider a pair $\mathbbm I_{k_n}, \mathbbm I_{k_{n+1}}$.
 There are two cases, $\rho_{k_n}<\rho_{k_{n+1}}$ and $\rho_{k_n}=\rho_{k_{n+1}}$. In the first case, we cannot have $\mathbbm I_{k_{n+1}}\subset \mathbbm I_{k_{n}}$. Thus, from (i) we have  $\mathbbm I_{k_{n}}\subset \mathbbm I_{k_{n+1}}$.  If $\rho_{k_n}=\rho_{k_{n+1}}$, then  $\mathbbm I_{k_n}\subsetsim\mathbbm I_{k_{n+1}}$ from part (ii).

The inclusion  $\mathbbm K_{k_n}\supset\mathbbm K_{k_{n+1}}$ is immediate from the definition.

(iv) Let $i\in \mathbbm I_{k_{n+1}}\smallsetminus \mathbbm I_{k_{n}}$. Then, since $i\not\in \mathbbm I_{k_{n}}$ $u^i_{k_n'} =u_m$ for $n'\leq n$. On the other hand, $\mu$-almost all $i\in I_{k_{n+1}}$ satisfy $i\in I_{k_{n'}}$, for $n'>n$. Thus, for $\mu$-almost all $i\in \mathbbm I_{k_{n+1}}\smallsetminus \mathbbm I_{k_{n}}$, the action $u^i$ is given by \eqref{Actions_}. The proof is similar for $i\in \mathbbm I_{k_1}$, and  $i\in [0,1)\smallsetminus\mathbbm I_{k_N} $.

\subsection{Proof of Proposition \ref{Reduced}}
\label{Reduced_proof}
If $\rho$ corresponds to a Nash equilibrium, then combining \eqref{NCP_large} and \eqref{compute_pi} we conclude that $H(\rho)=0$. Conversely, since all the terms of \eqref{Redu_Eq} are nonnegative, $H(\rho)=0$ implies that if $ \mu([\bar i_{j-1},\bar i_j)\cap[ \rho_{k_{n-1}},\rho_{k_{n}}))>0$, then $ F_{(j-1)2^N+n}(\tilde\pi(\rho))=\delta^j(\tilde\pi(\rho))$. Combining this with \eqref{compute_pi}, we conclude that  if for some $j,l$, $\pi_{(j-1)2^N+l}>0$ then  $  F_{(j-1)2^N+l}(\pi)=\delta^j(\pi)$, where $\pi=\tilde \pi(\rho)$. That is, $\pi$ is a Nash equilibrium.

From \eqref{compute_pi}, we observe that $\pi(\rho)$ is continuous with respect to $\rho$, since $\mu(\cdot)$ is the Lebesque measure. Moreover,   \eqref{Redu_Eq} can be expressed as:
\begin{equation*}
  H(\rho)=\pi(\rho)^T\Phi(\pi(\rho)).
\end{equation*}
The fact that $H(\rho)$ is nonnegative is a result of \eqref{NCP_large}. Furthermore, from Lemma \ref{l1}, $F_{(j-1)2^N+n}(\pi)=\tilde{J}_j(v^n,\pi)$ is continuous with respect to $\pi$. Additionally, $\delta^j(\pi)$, which is the minimum of $F_{(j-1)2^N+l}(\pi)=\tilde{J}_j(v^l,\pi)$ for all $v^l$, is continuous with respect to $\pi$ as the minimum of continuous functions. So, $\Phi(\pi)=F(\pi)-\delta(\pi)$ is continuous with respect to $\pi$ and $H(\rho)$ is continuous with respect to $\rho$ as composition of continuous functions.\\

 \bibliography{refs3}

\bibliographystyle{IEEEtran}

\end{document}